\documentclass[runningheads]{llncs}
\usepackage{xcolor}
\usepackage{lineno}
\usepackage{graphicx}
\usepackage{tabularx}
\usepackage{environ}

\makeatletter
\newcommand{\problemtitle}[1]{\gdef\@problemtitle{#1}}
\newcommand{\probleminput}[1]{\gdef\@probleminput{#1}}
\newcommand{\problemquestion}[1]{\gdef\@problemquestion{#1}}
\NewEnviron{problem2}{
  \problemtitle{}\probleminput{}\problemquestion{}
  \BODY
  \par\addvspace{.5\baselineskip}
  \noindent
  \begin{tabularx}{\textwidth}{@{\hspace{\parindent}} l X c}
    \multicolumn{2}{@{\hspace{\parindent}}l}{\@problemtitle} \\
    \textbf{Input:} & \@probleminput \\
    \textbf{Output:} & \@problemquestion
  \end{tabularx}
  \par\addvspace{.5\baselineskip}
}
\makeatother

\begin{document}
\title{Computing all-vs-all MEMs in grammar-compressed text\thanks{Supported by Academy of Finland (Grants 323233 and 339070), and by Basal Funds FB0001, Chile (first author)}}
%
%
\author{Diego D\'iaz-Dom\'inguez \and Leena Salmela}
\authorrunning{D. D\'iaz-Dominguez et al.}
%
\institute{Department of Computer Science, University of Helsinki, Finland
\email{\{diego.diaz,leena.salmela\}@helsinki.fi}}

\maketitle 

\begin{abstract} We describe a compression-aware method to compute all-vs-all maximal exact matches (MEM) among strings of a repetitive collection $\mathcal{T}$. The key concept in our work is the construction of a fully-balanced grammar $\mathcal{G}$ from $\mathcal{T}$ that meets a property that we call \emph{fix-free}: the expansions of the nonterminals that have the same height in the parse tree form a fix-free set (i.e., prefix-free and suffix-free). The fix-free property allows us to compute the MEMs of $\mathcal{T}$ incrementally over $\mathcal{G}$ using a standard suffix-tree-based MEM algorithm, which runs on a subset of grammar rules at a time and does not decompress nonterminals. By modifying the locally-consistent grammar of Christiansen et al.~\cite{christiansen2020optimal}, we show how we can build $\mathcal{G}$ from $\mathcal{T}$ in linear time and space. We also demonstrate that our MEM algorithm runs on top of $\mathcal{G}$ in $O(G +occ)$ time and uses $O(\log G(G+occ))$ bits, where $G$ is the grammar size, and $occ$ is the number of MEMs in $\mathcal{T}$. In the conclusions, we discuss how our idea can be modified to implement approximate pattern matching in compressed space.

\end{abstract}
\section{Introduction}

A \emph{maximal exact match} (MEM) between two strings is a match that cannot be extended without introducing mismatches or reaching an end in one of the strings. MEMs play an important role in biological sequence analyses~\cite{mummer,bowtie,li2013aligning} as they simplify finding long stretches of identical sequences. However, the rapid pace at which biological data have grown in later years has made the computation of MEMs intractable in these inputs. 

Seed-and-extend heuristics is a popular solution to scale the problem of approximate pattern matching in large collections~\cite{blast,blat,bowtie,li2013aligning,minimap2}. One of the aspects that impacts these heuristics' performance is the seeding mechanism. In this regard, using MEMs to seed alignments of near identical sequences, like genomes or proteins, offers two important benefits. First, it improves the accuracy of the result and reduces the cost of the heuristic's extension phase. The reason is because the result in an approximate alignment of highly-similar strings is indeed a sequence of long MEMs separated by small edits. Second, computing MEMs takes linear time~\cite{weiner1973linear,mccreight1976space}, and in small or medium-sized collections, it does not impose a considerable overhead. The problem, as mentioned before, is that genomic data have grown to a point where fitting the necessary data structures to detect MEMs into main memory is hardly possible.

State-of-the-art methods~\cite{sad07comp,ohl10cst,ros22mon,bou21pho,ros22fin,navcpmmem} address the problem of finding MEMs in large inputs by using compressed text indexes~\cite{g2018op,c2020gr} and matching statistics~\cite{cha94sub}. Although these approaches have demonstrated that string compression greatly reduces the overhead of computing MEMs in repetitive collections, they need to build the full index first, which can be expensive.

An efficient method that exploits text redundancies to compute all-vs-all MEMs in massive string collections could have important implications in Bioinformatics. Tasks like \emph{de novo} genome assembly, multiple genome alignment, or protein clustering could become available for inputs that are TBs in size. This result, in turn, could have a major impact on genomic research. 

\textbf{Our contribution} We present a method to compute all-vs-all MEMs in a collection $\mathcal{T}$ of repetitive strings. Our idea consists of building a context-free grammar $\mathcal{G}$ from $\mathcal{T}$, from which we compute the MEMs. Our grammar algorithm ensures a property that we call \emph{fix-free}, which means that the expansions of nonterminals with the same height form a set that is prefix-free and suffix-free. This idea enables the fast computation of MEMs by incrementally indexing parts of the grammar with simple data structures. Section 5 shows how we can build a fix-free grammar in linear time and space using a variant of the locally-consistent grammar of Christiansen et al.~\cite{christiansen2020optimal}. In Section 6, we describe how to compute a list $\mathcal{L}$ with the MEM ``precursors'' (prMEMs) from $\mathcal{G}$ in $O(G+|\mathcal{L}|)$ time and $O(\log G (G+|\mathcal{L}|))$ bits of space, where $G$ is the grammar size. In Section 6, we show how to enumerate all the $occ$ all-vs-all MEMs from $\mathcal{L}$ and $\mathcal{G}$, yielding thus an algorithm that runs in $O(G+occ)$ time and uses $O(\log G(G+occ))$ bits. 

\section{Preliminaries}

\subsubsection{String data structures}\label{sec:sa_st}
Consider a string $T[1..n-1]$ over the alphabet $\Sigma[2,\sigma]$, and the sentinel symbol $\Sigma[1]=\texttt{\$}$, which only occurs at $T[n]$. The \emph{suffix array}~\cite{MM93} of $T$ is a permutation $S\!A[1..n]$ that enumerates the suffixes $T[i..n]$ of $T$ in increasing lexicographic order, $T[S\!A[i]..n] < T[S\!A[i+1]..n]$, for $i \in [1..n-1]$. The \emph{longest common prefix} array $LCP[1..n]$~\cite{MM93} stores in $LCP[j]$ the longest common prefix between $T[S\!A[j-1]]$ and $T[S\!A[j]]$.
Given a vector $V[1..n]$ of integers, a \emph{minimum range query} $rmq(V, j, j')$, with $j<j'$, returns the minimum value in the arbitrary range $V[j..j']$. By encoding $LCP$ with support for $rmq$~\cite{fischermrq}, one can obtain the length of the prefix shared by two arbitrary suffixes $T[S\!A[j]..n]$ and $T[S\!A[j']..n]$ by performing $rmq(LCP, j, j')$.

\subsubsection{Locally-consistent parsing}\label{sec:lcp}
\emph{Parsing} consists in breaking a text $T[1..n]$ into a sequence of phrases. The parsing is \emph{locally-consistent}~\cite{cole1986deterministic} if, for any pattern $P$, its occurrences in $T$ are largely partitioned in the same way. There is more than one way to make a parsing locally consistent (see~\cite{sahinalp1994parallel,batu2006oblivious,christiansen2020optimal} for more details), but this work focuses on those using \emph{local minima}. A position $T[j]$ is a local minimum if $T[j-1]>T[j]<T[j+1]$. A method that uses this idea scans $T$, and for each pair of consecutive local minima $j$ and $j'$, it defines the phrase $T[j..j'-1]$. The procedure to compare adjacent positions in $T$ can vary. For instance, Christiansen et al.~\cite{christiansen2020optimal} first create a new string $\hat{T}$ where they replace equal-symbol runs by metasymbols. Then, they define a random permutation $\pi$ for the alphabet of $\hat{T}$, and compute the breaks as $\pi(\hat{T}[j-1])>\pi(\hat{T}[i])<\pi(\hat{T}[i+1])$. On the other hand, Nogueira et al.~\cite{n2018gr} compare consecutive suffixes rather than positions. Concretely, $T[j]$ is a local minimum if the suffix $T[j..n]$ is lexicographically smaller than the suffixes $T[j-1..n]$ and $T[j+1..n]$. This suffix-based local minimum was proposed by Nong et al.~\cite{n2009li} in their linear-time suffix array algorithm SAIS. They refer to it as an LMS-type position and to the phrases covering consecutive LMS-type positions as LMS-substrings.

\subsubsection{Grammar compression}
Grammar compression~\cite{ki2000gr,CLLPPSS05} is a form of lossless compression that encodes a string $T[1..n]$ as a context-free grammar $\mathcal{G}$ that only generates $T$. Formally, $\mathcal{G}=\{\Sigma, V, \mathcal{R}, S\}$ is a tuple of four elements, where $\Sigma$ is the set of terminals, $V$ is the set of nonterminals, $\mathcal{R}$ is a set of derivation rules in the form $X \rightarrow F$, with $X \in V$ being a nonterminal and $F \in (\Sigma \cup V)^{*}$ being its replacement, and $S \in V$ is the start symbol. In grammar compression, each nonterminal $X \in V$ appears only once on the left-hand sides of $\mathcal{R}$, which ensures the unambiguous decompression of $T$.

The graphical sequence of derivations that transforms $S$ into $T$ is referred to as the parse tree. The root of this tree is labelled $S$ and has $|A|$ children, with $S \rightarrow A$. The root's children are labelled from left to right according to $A$'s sequence. The subtrees for the root's children are recursively defined in the same way.
The height of a nonterminal $X$ is the longest path in the parse subtree induced by $X$'s recursive expansion. The grammar is said to be fully-balanced
if, for each nonterminal $X \in V$ at height $i$, the symbols in the right-hand side of $X \rightarrow A_1 A_2{\cdots}A_{x} \in \mathcal{R}$ are at height $i-1$.

The \emph{grammar tree} is a pruned version of the parse tree that, for each $X \in V$, keeps only the leftmost internal node labelled $X \in V$. The remaining internal nodes labelled $X$ are transformed into leaves. The leaves of the grammar tree induce a partition over $T$: for each grammar tree leaf $u$, its corresponding phrase is the substring in $T$ mapping the terminal symbols under the parse tree node from which $u$ was originated. One can classify the occurrences in $T$ of a pattern $P \in \Sigma^{*}$ into primary and secondary. A \emph{primary occurrence} of $P$ crosses two or more phrases in the partition induced by the grammar tree. On the other hand, a \emph{secondary occurrence} of $P$ is fully contained within a phrase.

\subsubsection{Locally-consistent grammar}\label{sec:lcg}

A grammar $\mathcal{G}$ generating $T[1..n]$ is locally consistent if the occurrences of the same pattern are largely compressed in the same way~\cite{jez2015approximation,christiansen2020optimal}. A mechanism to build $\mathcal{G}$ is by transforming $T$ in successive rounds of locally-consistent parsing~\cite{christiansen2020optimal,diaz21gram,n2018gr}. In every round $i$, the construction algorithm receives as input a string $T^{i}[1..n^{i}]$ over an alphabet $\Sigma^{i}$, which represents a partially-compressed version of $T$ (when $i=1$, $T^{1}=T$ and $\Sigma^{1}=\Sigma$). Then, it breaks $T^{i}$ using its local minima and creates a set $\mathcal{S}^{i}$ with the distinct phrases of the parsing. For every phrase $F \in \mathcal{S}^{i}$, the algorithm defines a new metasymbol $X$ and appends the new rule $X \rightarrow F$ into $\mathcal{R}$. The value of $X$ is the number of symbols in $\Sigma \cup V$ before the parsing round $i$ plus the rank of $F$ in an arbitrary order of $\mathcal{S}^{i}$. After creating the new rules, the algorithm replaces the phrases in $T^{i}$ with their corresponding metasymbols to produce another string $T^{i+1}[1..n^{i+1}]$, which is the input for the next round $i+1$. Notice that the alphabet $\Sigma^{i+1} \subset V$ for $T^{i+1}$ is the subset of metasymbols that the algorithm assigned to the phrases in $T^{i}$.  When $T^{i+1}$ does not have local minima, the algorithm creates the final rule $S \rightarrow T^{i+1}$ for the start symbol $S$ of $\mathcal{G}$ and finishes. The whole process runs in $O(n)$ time and produces a fully-balanced grammar. Christiansen at al.~\cite{christiansen2020optimal} showed that it is possible to obtain a locally-consistent grammar of size $O(\gamma \log n/\gamma)$ in $O(n)$ expected time, where $\gamma$ is the size of the smallest string attractor for $T$.

\section{Definitions}\label{sec:def}

Let $T[1..n]$ be a string of length $n$. We use the operator $T[1..j]$ to denote the $jth$ prefix of $T$ and the operator $T[j..]=T[j..n]$ to denote its $jth$ suffix. 

\begin{definition}
MEM: let $T_x[1..n_x]$ and $T_y[1..n_y]$ be two strings over the alphabet $\Sigma$. A maximal exact match $T_a[a..b]=T_y[a'..b']$ of length $\ell$, denoted $M(T_x, T_y, a , a', \ell)$, has the following properties: (i) $a=1$ or $a'=1$, or $T_x[a-1]\neq T_{y}[a'-1]$. (ii) $b=n_x$ or $b'=n_y$, or $T_x[b+1]\neq T_{y}[b'+1]$.
\end{definition}

We formulate the problem we address in this work as follows:

\begin{problem2}\label{def:memprob}
  \problemtitle{\textsc{AvAMem}}
  \probleminput{a string collection $\mathcal{T} = \{T_{1}, T_{2}, \ldots, T_{m}\}$ and an input integer $\tau$.}
  \problemquestion{every possible $M(T_x, T_y, a, a', \ell)$ with $T_{x}, T_{u} \in \mathcal{T}$ and $\ell \geq \tau$.}
\end{problem2}

We will refer to the parsing of Section~\ref{sec:lcp} as \textsc{LCPar}, and the grammar algorithm of Section~\ref{sec:lcg} as \textsc{LCGram}. Let $\mathcal{G}=(\Sigma, V, \mathcal{R}, S)$ be the fully-balanced grammar constructed by \textsc{LCGram} from $\mathcal{T}$ in $h$ rounds of parsing. We assume for the moment that the definition of local minima is arbitrary, but sequence-based. We denote the sum of the right-hand side lengths of $\mathcal{R}$ as $G$ and the number of nonterminals as $g=|V|$. Additionally, we consider the partition $V=\{V^{1},\ldots,V^{h}\}$ such that every $V^{i}$, with $i \in [1..h]$, has the nonterminals generated during the parsing round $i$. Similarly, the partition $\mathcal{R}=\{\mathcal{R}^{1}, \ldots, R^{h}\}$ groups in $\mathcal{R}^{i}$ the rules generated during the parsing round $i$. We denote as $G^{i}$ the sum of the right-hand side lengths of $\mathcal{R}^{i}$ and $g^{i}=|V^{i}|$. We will refer to the sequence $[1..h]$ as the \emph{levels} of the grammar, which is read bottom-up in the parse tree. $\mathcal{T}$ is at level $0$.

We assume \textsc{LCGram} compresses the strings of $\mathcal{T}$ independently but collapses all the rules in one single grammar $\mathcal{G}$. Thus, the rule $S \rightarrow A_1..A_{m'}$ encodes the compressed strings of $\mathcal{T}$ concatenated in the string $A_1 \ldots A_{m'}$.   

The operator $exp(X)$ returns the string in $\Sigma^{*}$ resulting from the recursive expansion of $X \in V$. The function $exp$ can also receive as input the sequence $X_{1}{\cdots}X_{b} \in V^{*}$, in which case it returns the concatenation $exp(X_1){\cdots}exp(X_b)$. The operation $lcp^{i}(X, Y)$ receives two nonterminals $(X, Y) \in V^{i}$ and returns the longest common prefix between $exp(X)$ and $exp(Y)$. Similarly, the operator $lcs^{i}(X, Y)$ returns the longest common suffix between $exp(X)$ and $exp(Y)$.

\begin{definition}\label{def:ff}
A set $\mathcal{S}$ of strings is fix-free iff, for any pair $F, Q \in \mathcal{S}$, the string $F$ is not a suffix nor a prefix of $Q$, and vice-versa.
\end{definition}

\begin{definition}\label{def:ffg}
A grammar $\mathcal{G}$ is fix-free iff it is fully balanced, and for any level $i$, the set $\mathcal{S}^{i}=\{exp(X_{1}), \ldots, exp(X_{g^{i}})\}$, with $X_{1}, \ldots ,X_{g^{i}} \in V^{i}$, is fix-free.
\end{definition}

\begin{definition}\label{def:po_MEM}
Primary MEM (prMEM): let $M(T_{x}, T_{y}, a, 'a, \ell)$, with $T_x, T_{y} \in \mathcal{T}$, be a MEM. $M(T_{x}, T_{y}, a, 'a, \ell)$ is primary if both $T_{x}[a..a+\ell-1]$ and $T_{y}[a'..a'+\ell-1]$ are primary occurrences of the pattern $T_{x}[a..a+\ell-1]=T_{y}[a'..a'+\ell-1]$.
\end{definition}

\section{Overview of our algorithm}

We solve $\textsc{AvAMem}(\mathcal{T}, \tau)$ in three steps: (i) we build a fix-free grammar from $\mathcal{T}$ using a modification of \textsc{LCGram}, (ii) we compute a list $\mathcal{L}$ storing the prMEMs of $\mathcal{T}$, and (iii) we use $\mathcal{L}$ and $\mathcal{G}$ to report the positions in $\mathcal{T}$ of all the MEMs. 

The advantage of $\mathcal{G}$ being fix-free is that we can use the following lemma:

\begin{lemma}\label{lem:pomem}
Let $\mathcal{G}$ be a fix-free grammar. Two rules $X \rightarrow AZB,\ Y \rightarrow CZD \in \mathcal{R}^{i}$, with $Z \in V^{i-1*}$ and $A,B,C,D \in V^{i-1}$ yield a prMEM if $\ell = lcs^{i-1}(A, C) + exp(Z) + lcp^{i-1}(B,D) \geq \tau$.
\end{lemma}

\begin{proof}
Both $exp(X)$ and $exp(Y)$ contain $exp(Z)$ as a substring as the rules for $X$ and $Y$ have occurrences of $Z$, and there is only one string $exp(Z)$ the grammar can produce. Still, the strings $exp(A)$ and $exp(C)$ (respectively, $exp(B)$ and $exp(D)$) could share a suffix (respectively, a prefix), meaning that the prMEM extends to the left of $Z$ (respectively, the right of $Z$). The fix-free property guarantees that the left boundary of the prMEM lies at some index in the right-to-left comparison of $exp(A)$ and $exp(C)$, and that the right boundary lies at some position within the left-to-right comparison of $exp(B)$ and $exp(D)$.
\end{proof}

Lemma~\ref{lem:pomem} offers a simple solution to detect prMEMs as we do not have to look into other parts of $\mathcal{G}$'s parse tree to check the boundaries of the prMEM in $(X, Y)$. The only remaining aspects to consider are, first, how to get a fix-free grammar, and then, how to compute $lcs(A, C)$, $lcp(B, D)$, and $|exp(Z)|$ efficiently. Once we solve these problems, finding prMEMs reduces to run a suffix-tree-based MEM algorithm over the right-hand sides of each $\mathcal{R}^{i}$. 
On the other hand, getting the positions in $\mathcal{T}$ of the MEMs (step 3.) requires traversing the rules of $\mathcal{G}$, but it is not necessary to perform any string comparison.

\section{Building the fix-free grammar}

We first describe the parsing we will use in our variant of \textsc{LCGram}. 
We refer to this procedure as \textsc{FFPar}. The input is a string collection $\mathcal{T}'$ over the alphabet $\Sigma' = \{\texttt{\$}\} \cup \Sigma \cup \{\texttt{\#}\}$, where each string $T_x=\texttt{\$}T\texttt{\$}\texttt{\#} \in \mathcal{T}'$ is flanked by the symbols $\{\texttt{\$}, \texttt{\#}\} \notin \Sigma$ that do not occur in the internal substring $T \in \Sigma^{*}$.

We choose a function $h: \Sigma' \rightarrow [0..p+1]$ that maps symbols in $\Sigma$ to integers in the range $[1..p]$ uniformly at random, where $p> |\Sigma|$ is a large prime number. If $c \in \Sigma$, then $h(c)= (ac+b) \bmod p$, where $a,b \in [1..p]$ are chosen uniformly at random. If $c=\texttt{\$} \notin \Sigma$, then $h(c)=0$, and if $c=\texttt{\#}$, then $h(c)=p+1$. 

We use $h$ to define the local minima in $\mathcal{T}'$. The idea is to combine $h$ with a scheme to classify symbols similar to that of SAIS~\cite{n2009li}.
Let $T_{x}[1..n_x]$ be a string in $\mathcal{T}'$. A position $T_{x}[j]$ with $j \in [2..n_x-1]$ has three possible classifications:

\begin{itemize}
\item L-type : $h(T_{x}[j])>h(T_{x}[j+1])$ or $T_{x}[j]=T_{x}[j+1]$ and $T_{x}[j+1]$ is L-type.
\item S-type : $h(T_{x}[j])<h(T_{x}[j+1])$ or $T_{x}[j]=T_{x}[j+1]$ and $T_{x}[j+1]$ is S-type.
\item LMS-type : $T_{x}[j]$ is S-type and $T_{x}[j-1]$ is L-type.
\end{itemize}

The LMS-type positions are the local minima of $\mathcal{T}'$. Randomising the local minimum definition aims to protect us against adversarial inputs. Notice that $h$ is not a random permutation $\pi$ like Christiansen et al.~\cite{christiansen2020optimal}, but it still defines a total order over $\Sigma$ as it does not assign the same random value to two symbols, which is enough for our purposes. 

\begin{lemma}\label{lem:ffpar}
\textsc{FFPar}: let $T_{x}[1..n_x] \in \mathcal{T}'$ be a string in the collection. For every pair of consecutive local minima $j>1$ and $j'<n_x$, we define the phrase $T_{x}[j-1..j'+1]$. For the leftmost local minimum $T_{x}[j]$, we define the phrase $T_{x}[1..j+1]$. For the rightmost local minimum $T_x[j']$, we define $T_x[j'-1..n_x]$. The resulting set of parsing phrases is fix-free.
\end{lemma}

\begin{proof}
Let us first consider the set $\mathcal{S}$ created by \textsc{LCPar}. We will say that a phrase $T_{x}[j..j-1]$ in \textsc{LCPar} is an \emph{instance} of $F \in \mathcal{S}$ if $T_{x}[j..j'-1]$ matches $F$. When $F$ occurs in $\mathcal{T}'$ as a proper suffix or prefix of another phrase in $\mathcal{S}$, it is not an instance, only an occurrence of $F$.

Let $\mathcal{W} \subset \mathcal{S}$ be a subset of phrases and let $F \in \mathcal{S} \setminus \mathcal{W}$ be a phrase occurring as a proper prefix in each element of $\mathcal{W}$. Consider any pair of instances $F=T_{x}[j..j'-1]$ and $Q=T_{y}[l..l'-1] \in \mathcal{W}$ such that $T_{x},T_{y} \in \mathcal{T}'$, and $(l,l')$ (respectively, $(j,j')$) are consecutive local minima. The substring $T_{x}[j'..j'+1]$ following $F$'s instance cannot be equal to the substring $T_{y}[l+|F|..l+|F|+1]$ following $F$'s occurrence within $Q$ because $T_{x}[j']$ is a local minimum while $T_{y}[l+|F|]$ is not. We know that $T_{y}[l+|F|]$ is within $Q$ because $T_{y}[l..l+|F|-1]$ is an occurrence of $F$ that is a proper prefix of $Q$. Therefore, if $T_{y}[l+|F|]$ were a local minimum, $Q$ would also be an instance of $F$. In conclusion, running \textsc{FFPar} will right-extend every instance of $F$ in $\mathcal{T}'$ such that none of the resulting phrases is a proper prefix in the right-extended phrases of $\mathcal{W}$.

We develop a similar argument for the left extension. Suppose, in this case, $F \in \mathcal{S} \setminus \mathcal{W}$ is a proper suffix in $\mathcal{W} \subset \mathcal{S}$. As before, we assume $\mathcal{S}$ was built using \textsc{LCPar} and that $F=T_{x}[j..j'-1]$ and $Q=T_{y}[l..l'-1] \in \mathcal{W}$ are phrase instances. The symbol $T_{x}[j-1]$ preceding $F$'s instance cannot be equal to the symbol $T_{y}[l'-|F|-1]$ preceding the occurrence $F=T_{y}[l'-|F|..l'-1]$ within $Q$. The reason is because $F[1]=T_{x}[j]$ is a local minimum, while $Q[|Q|-|F|+1]=T_{y}[l'-|F|]=F[1]$ is not. In conclusion, \textsc{FFPar} left-extends each occurrence of $F$ such that none of the resulting phrases is a proper suffix in the left-extended elements of $\mathcal{W}$.
\end{proof}

\subsubsection{The \textsc{FFGram} algorithm}

\textsc{FFGram} is our \textsc{LCGram} variant that builds a fix-free grammar by applying successive rounds of \textsc{FFPar}  (Definition~\ref{lem:ffpar}). The input of \textsc{FFGram} is the collection $\mathcal{T}^{1}$ built from $\mathcal{T}$ by adding the special flanking symbols $\texttt{\$}T_{x}\texttt{\$}\texttt{\#}$ in every $T_{x} \in \mathcal{T}$. The alphabet of $\mathcal{T}^{1}$ is $\Sigma^{1}=\{\texttt{\$}\} \cup \Sigma \cup \{\texttt{\#}\}$. 

In every round $i$, we create a random hash function $h^{i} : \Sigma^{i} \rightarrow [0..p+1]$, with $p>|\Sigma^{i}|$, to define the local minima of $\mathcal{T}^{i}$. Then, we run \textsc{FFPar} using $h^{i}$ and sort the resulting set $\mathcal{S}^{i}$ of phrases in lexicographical order starting from the leftmost proper suffix of each string. The relative order of strings of $\mathcal{S}^{i}$ differing only in the leftmost symbol is arbitrary. Let $c$ be the number of symbols in $\Sigma \cup V$ before parsing round $i$ and let $r$ the rank of $F \in \mathcal{S}^{i}$ in the ordering we just defined. We assign the metasymbol $X=c+r \in V^{i}$ to $F$ and append the new rule $X \rightarrow F\in \mathcal{R}^{i}$. The last step in the round is to create the collection $\mathcal{T}^{i+1}$ by replacing the phrases with their corresponding metasymbols. Additionally, we define two special new symbols $\{\$, \#\}$, which we append at the ends of the strings in $\mathcal{T}^{i+1}$. Thus, the alphabet of $\mathcal{T}^{i}$ becomes $\Sigma^{i} = \{\$\} \cup V^{i} \cup \{\#\}$. We assume the occurrences of the special symbols $\#,\$$ in the right-hand sides of $\mathcal{R}$ expand to the empty string. \textsc{FFGram} stops after $h$ parsing rounds, when the input $\mathcal{T}^{i}$ does not have local minima. Figure~\ref{fig:ovp_gramp} shows an example. 
 
\begin{figure}[t]
\centering
\includegraphics[width=0.75\textwidth]{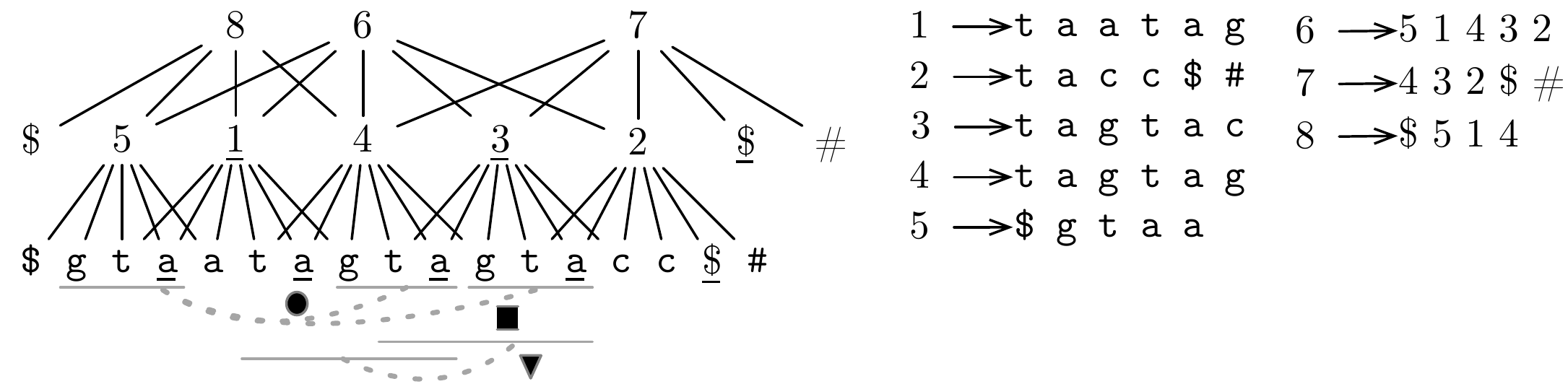}
\caption{Execution of \textsc{FFGram} on \texttt{\$gtaatagtagtacc\$\#}. The left side shows the parse tree up to level $2$, and the right side shows the corresponding rules for those levels. The underlined symbols are local minima. The horizontal lines in grey are MEMs of length $\geq \tau = 3$. Each shape identifies a specific MEM in Figures~\ref{fig:prmem_it1} and \ref{fig:prmem_it2}.}
\label{fig:ovp_gramp}
\end{figure}

\begin{lemma}\label{lem:ffg}
The grammar $\mathcal{G}$ resulting from running \textsc{FFGram} over $\mathcal{T}$ is fix-free.
\end{lemma}

\begin{proof}
Consider the execution of \textsc{FFPar} over $\mathcal{T}^{1}$ in the first round of \textsc{FFGram}. The output set $\mathcal{S}^{1}$ is over the alphabet $\{\texttt{\$}\} \cup \Sigma \cup \{\texttt{\#}\}$ and is, by Lemma~\ref{lem:ffpar}, fix-free. Thus, the symbols of $V^{1}$ meet the fix-free property of Definition~\ref{def:ffg}. Now consider the parsing rounds $i-1$ and $i\geq 2$. Assume without loss of generality that the nonterminals in $V^{i-1}$ meet Definition~\ref{def:ffg}. The recursive definition of \textsc{FFGram} implies that the phrases in $\mathcal{S}^{i}$ are over the alphabet $V^{i-1}$. Thus, for any pair of different strings $F,Q \in \mathcal{S}^{i}$ sharing a prefix $F[1..q]=Q[1..q]$, the symbols $F[q+1]\neq E[q+1]$ expand to different sequences $exp(F[q+1])\neq exp(Q[q+1])$ that are not prefix one of the other as $F[f+1],Q[q+1]$ belong to $V^{i-1}$. Therefore, $exp(F)$ and $exp(Q)$ do not prefix one to the other either. The same argument applies when $F$ and $Q$ share a prefix. We conclude then that the metasymbols of $V^{i}$ also meet the fix-free property of Definition~\ref{def:ffg}.   
\end{proof}

\subsubsection{Overlaps in the fix-free grammar}

A relevant feature of \textsc{FFGram} is that consecutive nonterminals in the right-hand sides of $\mathcal{R}$ cover overlapping substrings of $\mathcal{T}$. This property allows us to compute prMEMs as described in Lemma~\ref{lem:pomem}. The downside, however, is that expanding substrings of $\mathcal{T}$ from $\mathcal{G}$ is now more difficult. However, the number of symbols overlapping in every grammar level is constant (one to the left and two to the right). Depending on the situation, we might want to decompress nonterminals considering or skipping overlaps. Thus, we define the operations $efexp(X)$, $lexp(X)$, and $rexp(X)$ that return different types of nonterminal expansions.
These functions only differ in the edges they skip in $X$'s subtree during the decompression.

\begin{itemize}
\item $efexp(X)$: recursively skips the leftmost and two rightmost edges.
\item $lexp(X)$: recursively skips the two rightmost edges, and the leftmost edges when the parent node does not belong to the leftmost branch.
\item $rexp(X)$ recursively skips the leftmost edges, and the two rightmost edges when the parent node does not belong to the rightmost branch. 
\end{itemize}

Figure \ref{fig:ovp_exp} shows examples of $efexp(X)$, $lexp(X)$ and $rexp(X)$. We also modify the functions $lcs^{i}$ and $lcp^{i}$ described in Definitions~\ref{sec:def}. Let $X,Y \in V^{i}$ be two nonterminals at level $i$. The function $lcs^{i}(X, Y)$ now returns the longest common suffix between $lexp(X)$ and $lexp(Y)$; and $lcp^{i}(X, Y)$ returns the longest common prefix between $rexp(X)$ and $rexp(Y)$. 

We remark that $efexp$, $lexp$, $rexp$, $lcs^{i}$, and $lcp^{i}$ are virtual as our algorithm to find prMEMs never calls them directly. Instead, it incrementally produces satellite data structures with precomputed answers to solve them in $O(1)$ time.

\section{Computing prMEMs in the fix-free grammar}\label{sec:prmem}

Our MEM algorithm uses the grammar $\mathcal{G}=\{\Sigma, V, \mathcal{R}, S\}$ resulted from running \textsc{FFGram} with $\mathcal{T}'$, and produces the list $\mathcal{L}$ of prMEMs in $\mathcal{T}$. We consider a set $\mathcal{O}$ of $g=|\mathcal{R}|$ rules storing the cumulative lengths of the $efexp$ expansions for the right-hand sides of $\mathcal{R}$. We also define a logical partition for $\mathcal{O}$ according to the grammar levels. Let $X \rightarrow A_1A_2{\cdots}A_{x} \in \mathcal{R}^{i}$ be a rule at level $i$. The rule $X \rightarrow c_1 \cdots c_{x} \in \mathcal{O}^{i}$ stores in $c_{j}$, with $j \in [3..x-1]$, the value $c_{j} = efexp(A_{2})+\cdots+efexp(A_{j-1})$. To avoid recursive overlaps, we set $c_1=0, c_2=0$, and $c_{x-1}=c_{x-2}, c_{x}=c_{x-2}$. The leftmost tree of Figure~\ref{fig:ovp_exp} shows an example of a rule in $\mathcal{O}$. We assume \textsc{FFGram} already constructed $\mathcal{O}$.  

\begin{figure}[t]
\centering
\includegraphics[width=0.9\textwidth]{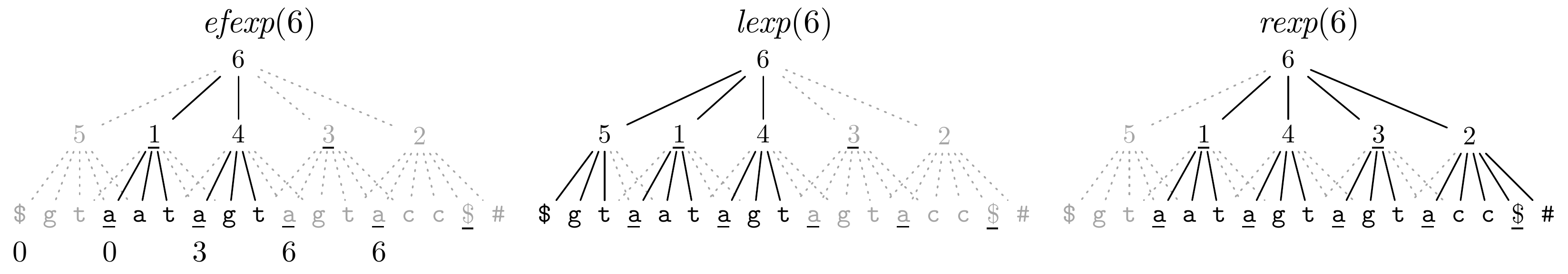}
\caption{Expansions for nonterminal $6$ of Figure~\ref{fig:ovp_gramp}. Dashed lines are skipped branches. The sequence of numbers below $efexp(6)$ corresponds to the rule $6 \rightarrow 0\ 0\ 3\ 6\ 6 \in \mathcal{O}$.}
\label{fig:ovp_exp}
\end{figure}

\paragraph{prMEM encoding.} Every element of $\mathcal{L}$ is a tuple $(X, Y, o_X, o_Y, \ell)$ of four elements. $X$ and $Y$ are the nonterminals labelling the lowest nodes in $\mathcal{G}$'s grammar tree that encode the primary occurrences for the MEM's  sequence. The fields $o_X$ and $o_Y$ are the number of terminal symbols preceding the prMEM within $efexp(X)$ and $efexp(Y)$ (respectively), and $\ell$ is the length of the prMEM.

\paragraph{Grammar encoding.} We encode every subset $\mathcal{R}^{i} \subset \mathcal{R}$ as an individual string collection concatenated in one single array $R^{i}[1..G^{i}]$. We create an equivalent array $O^{i}$ for the rules of $\mathcal{O}^{i} \subset \mathcal{O}$. We also consider a function $map(X)=j$ that indicates that the right-hand side $F=R^{i}[a..a']$ of $X\rightarrow F \in \mathcal{R}^{i}$ is the $jth$ string of $R^{i}$ from left to right. Additionally, we define the function $parent(b)=X$ that returns $X$ for each position $b \in [a..a']$. We assume $map$ and $parent$ are implemented in $O(1)$ time and $G_{i}+o(G_{i})$ bits using bit vectors.

Our prMEM algorithm is an iterative process that, each step $i$, searches for prMEMs in the rules of the grammar level $i$. Still, there is a slight difference between the iteration $i=1$ and the others $i>1$, so we explain them separately. 

\subsubsection{First iteration}\label{sec:f_it}

The first step in iteration $i=1$ is to create a sparse suffix array $S\!A$ for $R^{1}$ that discards each position $S\!A[j]$ meeting one of the following conditions: (i) $R^{i}[S\!A[j]]$ is the start of a phrase, $R^{i}[S\!A[j]+1]$ is the end of a phrase, or (iii) $R^{i}[S\!A[j]]$ is the end of a phrase. We refer to the resulting sparse suffix array as $A$. The next step is to produce the LCP array for $A$, which we name $GLCP$ for convenience. We then run the suffix-tree-based MEM algorithm~\cite{weiner1973linear} using $A$ and $GLCP$. We implement this step by simulating with $GLCP$ a traversal over the compact trie induced by the suffixes of $R^{i}$ in $A$ (see~\cite{lcplinear,enhsa} for the traversal). Every time the MEM algorithm reports a triplet $(A[u], A[u'], l)$ for a MEM $R^{i}[A[u]..A[u]+l-1]=R^{i}[A[u']..A[u']+l-1]$, we append the tuple $(X=parent(A[u])$, $Y=parent(A[u'])$, $o_X=O^{i}[A[u]]$, $o_Y=O^{i}[A[u']]$, $\ell=l$) into $\mathcal{L}$. Figure~\ref{fig:prmem_it1} exemplifies this algorithm. Once we finish running the MEM algorithm, we obtain satellite data structures for the next iteration $i+1$:

\begin{itemize}

\item A vector $P^{1}[1..g^1]$ encoding the permutation of $R^{1}$ resulted from sorting the $lexp$ expansions of the strings in colexicographical order.

\item A vector $LCS^{1}[1..g^{1}]$ storing LCS (longest common suffix) values between the $lexp$ expansions of strings in $R^{i}$ that are consecutive in the permutation $P^{1}$.
We encode $LCS^{1}$ with support for range minimum queries in $O(1)$ time. Thus, given two nonterminals $X,Y \in V^{1}$, we implement $lcs^{1}(X, Y)$ as $rmq(LCS^{1}, P^{1}[map(X)], P^{1}[map(Y)])$.

\item A vector $LCP^{1}[1..g^{1}]$ storing LCP values between the $rexp$ expansions of consecutive strings of $R^{1}$.
$LCP^{1}$ also supports $O(1)$-time $rmq$ queries so we implement $lcs^{1}(X, Y)$ as $rmq(LCP^{1}, X, Y)$. 
\end{itemize}


\subsubsection{Next iterations}

For $i>1$, we receive as input the tuple $(R^{i}, O^{i})$ and the vectors $P^{i-1}$, $LCS^{i-1}$, $LCP^{i-1}$. We assume $LCS^{i-1}$ and $LCP^{i-1}$ support $rmq$ queries in $O(1)$ time so we can implement $lcs^{i-1}$ and $lcp^{i-1}$ in $O(1)$ time as well. 

We compute $A$ and $GLCP$ for $R^{i}$ as in Section~\ref{sec:f_it}. However, we transform $GLCP$ to store the LCP values of the $rexp$ expansions of the suffixes of $R^{i}$ in $A$. Let $R^{i}[A[j]..]$ and $R^{i}[A[j+1]..]$ be two consecutive suffixes in $A$ sharing a prefix of length $GLCP[j+1]=l$. We update this value to $GLCP[j+1]=O^{i}[A[j+1]+l] + lcp^{i-1}(R^{i}[A[j]+l]), R^{i}[A[j+1]+l])$.

We use $GLCP$ and $A$ to simulate a traversal over the compact trie induced by the $rexp$ expansions of the $R^{i}$ suffixes in $A$. The purpose of the traversal is, again, to run the suffix-tree-based MEM algorithm. Every time this procedure reports a triplet $(A[u], A[u'], l)$ as a MEM, we compute $o=lcs(R^{i}[u-1], R^{i}[u'-1])$, and insert the tuple $(X=parent(A[u]), Y=parent(A[u'])), o_X=O^{i}[A[u]]-o+1, o_Y=O^{i}[A[u']]-o+1, \ell=o+l)$ into $\mathcal{L}$. Figure~\ref{fig:prmem_it2} shows an example. The final step in the iteration is to produce the satellite data structures:

\begin{itemize}
\item $P^{i}[1..g^{i}]$: we sort the strings in $R^{i}$ colexicographically using $P^{i-1}$. We define the relative order of any pair of strings $F, Q \in R^{i}$ by comparing their sequences $P^{i-1}[F[1]]{\cdots}P^{i-1}[F[|F|-2]]$ and $P^{i-1}[Q[1]]{\cdots}P^{i-1}[Q[|Q|-2]]$ from right to left. The resulting permutation $P^{i}$ has the following property: let $X, X' \in V^{i}$ be two nonterminals. If $P^{i}[map(X)]<P^{i}[map(X')]$, it means $lexp(X)$ is colexicographically equal or smaller than $lexp(X')$.

\item $LCS^{i}[1..g^{i}]$: we scan the strings of $R^{i}$ in $P^{i}$ order. Let $R^{i}[a..a']$ and $R^{i}[b..b']$ be two consecutive strings in the permutation of $P^{i}$. That is, $X=parent(a)$ and $X'=parent(b)$ such that $P^{i}[map(X)]=j$ and $P^{i}[map(X')]=j+1$. Assume their prefixes $R^{i}[a..a'-2]$ and $R^{i}[b..b'-2]$ share a suffix of length $l\geq 0$. We set $LCS^{i}[j+1]= O^{i}[a']-O^{i}[a'-l-1] + lcs^{i-1}(R^{i}[a'-l-2], R^{i}[b'-l-2])$.

\item $LCP^{i}[1..g^{i}]$: we scan the strings of $R^{i}$ from left to right. Let $R^{i}[a..a']$ and $R^{i}[b..b']$ be two strings in $R^{i}$ with $X=parent(a)$ and $X'=parent(b)=X+1$. Assume their suffixes $R^{i}[a+1..a']$ and $R^{i}[b+1..b']$ share a prefix of length $l\geq 0$. We set $LCS^{i}[map(X')]= O^{i}[b+l+1] + lcp^{i-1}(R^{i}[a+1+l], R^{i}[b+1+l])$.
\end{itemize}


\begin{theorem}
Let $\mathcal{G}=\{\Sigma, V, \mathcal{R}, S\}$ be a fix-free grammar of size $G$ built with \textsc{FFGram} using the collection $\mathcal{T}=\{T_1, \ldots, T_{u}\}$. It is possible to obtain from $\mathcal{G}$ the list $\mathcal{L}$ with the prMEMs of $\mathcal{T}$ in $O(G+|\mathcal{L}|)$ time and $O(\log G(G + |\mathcal{L}|))$ bits. 
\end{theorem}


\begin{proof}
In iteration $i=1$, constructing $A$ and $GLCP$ takes $O(G_{1})$ time and $O(G^{1}\log G^{1})$ bits~\cite{n2009li,lcplinear}. Then, we implement the suffix-tree-based algorithm to report MEMs in $R^{i}$ by combining the method of Abouelhoda et al.~\cite{enhsa}, that visits the nodes of the compact trie induced by $GLCP$ in $O(G^{i})$ time, with Lemma 11.4 of M\"akinen et al.~\cite{makinen2015genome}, which reports the MEMs of every internal node. These ideas combined take $O(G^1 + e^1)$ time and $O(\log G^{1}(G^{1}+e^1))$ bits, where $e^{1}$ is the number of prMEMs in the grammar level $1$. The final step is to build $LCP^{1}, LCS^{1}$, and $P^{1}$. \textsc{FFGram} sorted the strings of $R^{1}$ in lexicographical order according to their $rexp$ expansions, so building $LCP^{1}$ reduces to scan $R^{1}$ from left to right and compute LCP values between consecutive elements. This process takes $O(G^{1})$ time and $O(G^{1}\log G^{1})$ bits. Then, we obtain $P^{1}$ by running SAIS in the reversed strings of $R^{1}$, which also takes $O(G^{1})$ time. We use the reversed strings to build $LCS^{1}$ as we did with $LCP^{1}$. Finally, giving $rmq$ support to $LCP^{1}$ and $LCP^{1}$ takes $O(G^{1})$ time and $O(G^{1}\log G^{1})$ bits if we use the data structure of~Johannes Fisher~\cite{fischermrq}. Summing up, the iteration $i$ runs in $O(G^{1}+e^1)$ time and uses $O(\log G^{1}(G^{1} + e^1))$ bits. Each iteration $i>1$ performs the same operations, but it also updates $GLCP$, $LCP^{i}$, and $LCS^{i}$. These updates require linear scans of the arrays as processing each position performs $O(1)$ access to $O^{i}$ and $O(1)$ calls to $lcs^{i-1}$ or $lcp^{i-1}$, which we implement in $O(1)$ time. Thus, the cost of iteration $i$ is $O(G^{i}+e^i)$ time and $O(\log G^{i}(G^{i} + e^i))$ bits, where $e^i$ is the number prMEMs in the grammar level $i$. Combining the $h$ grammar levels, the cost to compute $\mathcal{L}$ is $O(G+|\mathcal{L}|)$ time and $O(\log G(G+|\mathcal{L}|))$ bits. 

\end{proof}

\section{Positioning MEMs in the text}

The last aspect we cover to solve $\textsc{AvAMEM}(\mathcal{T}, \tau)$ is computing from ($\mathcal{L}, \mathcal{G}$) the positions in $\mathcal{T}$ of the MEMs. We assume that the collections $\{R^{1},\ldots, R^{h}\}$ of Section~\ref{sec:prmem} are concatenated in one single array $R[1..G]$, and that the collections $\{O^{1},\ldots, O^{h}\}$ are concatenated in another array $O[1..G]$. We define the function $stringid$, which takes as input an index $u$ within $R$ with $S=parent(u)$ (start symbol in $\mathcal{G}$), and returns the identifier of the string of $\mathcal{T}$ where $R[u]$ lies. 

We first simplify $\mathcal{G}$ to remove the unary paths in its parse tree (see Section~\ref{sec:simp_gram}). This extra step is not strictly necessary, but it is convenient to avoid redundant work.
After the simplification, we create an array $N[1..G]$, which we divide into $g$ buckets. Each bucket $b \in [1..g]$ stores in an arbitrary order the positions in $R$ for the occurrences of $b \in \Sigma \cup V$. Additionally, we create a vector $C[1..g]$ that stores in $C[b]$ the position in $N$ where the bucket for $b \in \Sigma \cup V$ starts.

We report MEMs as follows: we insert the tuples of $\mathcal{L}$ into a stack. Then, we extract the tuple $(X, Y, o_Y, o_X, \ell)$ from the top of the stack and compute $s_X=C[X], e_X=C[X+1]-1$ and $s_Y=C[Y], e_Y=C[Y+1]-1$. For every pair of indexes $(u,u') \in [s_X, e_X] \times [s_Y, e_Y]$, we get $X'=parent(N[u])$, and if $X'=S$, we set $X'=stringid(N[u])$. We do the same with $N[u']$ and store the result in a variable $Y'$. Now we produce the new tuple $(X', Y', o_{X'}=O[N[u]]+o_X, o_{Y'}=O[N[u']]+o_Y, \ell)$. If both $X'$ and $Y'$ are strings identifiers, we report the tuple $M(X',Y',o_{X'}, o_{Y'}, \ell)$ as an output of \textsc{AvAMEM}, otherwise we insert the new tuple into the stack. If $X'$ or $Y'$ is a string identifier rather than a nonterminal, we flag the tuple to indicate that one of the elements is a string. Thus, when we visit the tuple again, we avoid recomputing its values. The report of MEMs ends when the stack becomes empty. 

\begin{theorem}
Let $\mathcal{G}$ be a fix-free grammar of size $G$ constructed with \textsc{FFGram} using the collection $\mathcal{T}=\{T_1, \ldots, T_{u}\}$, and let $\mathcal{L}$ be the list of prMEMs in $\mathcal{G}$ with length $>\tau$, where $\tau$ is an input parameter. Given the simplified version of $\mathcal{G}$ and $\mathcal{L}$, it is possible to report the positions in $\mathcal{T}$ of the $occ$ MEMs of length $>\tau$ in $O(G+occ)$ time and $O(\log G(G + occ))$ bits.
\end{theorem}

\begin{proof}
The $G$ term in the time complexity comes from the grammar simplification and the construction of $N$ and $C$. Let $(X, Y, o_X, o_Y,\ell) \in \mathcal{L}$ be a prMEM whose sequence in $\mathcal{T}$ is $L$, with $|L|=\ell$. Let us assume $aLb$ is the primary occurrence of $L$ under $X$ and $xLz$ is the primary occurrence under $Y$. In our algorithm, the access pattern in $N$ simulates a bottom-up traversal of $\mathcal{G}$'s grammar tree that visits every node labelled $X'$ such that $exp(X')$ has $aLb$, and every node labelled $Y'$ such that $exp(Y')$ has $xLz$. Our idea is similar to reporting secondary occurrences in the grammar self-index of Claude and Navarro~\cite{cn2012im}. They showed that the cost of traversing the grammar to enumerate the $occ_{P}$ occurrences of a pattern $P$ amortizes to $O(occ_{P})$ time. The argument is that, in a simplified $\mathcal{G}$, each node we visit in the grammar tree yields at least one occurrence in $\mathcal{T}$. However, our traversal processes two patterns at the same time ($aLb$ and $yLz$), pairing each occurrence of one with each occurrence of the other. Thus, our amortized time to process a prMEM tuple is $O(occ_X \times occ_Y)$, where $occ_X$ and $occ_Y$ are the numbers of occurrences in $\mathcal{T}$ for $aLb$ and $yLz$, respectively. Summing up, the cost of processing all the tuples in $\mathcal{L}$ is $O(occ)$ time.
\end{proof}

\begin{corollary}
Let $\mathcal{T}$ be a string collection of $n$ symbols containing $occ$ MEMs of length $\geq \tau$, $\tau$ being a parameter. We can solve $\textsc{AvAMEM}(\mathcal{T}, \tau)$ by building a fix-free grammar $\mathcal{G}$ of size $G$ in $O(n)$ time and $O(G\log G)$ bits, and then computing the $occ$ MEMs of $\mathcal{T}$ over $\mathcal{G}$ in $O(G+occ)$ time and $O(\log G(G+occ))$ bits.
\end{corollary}

\section{Concluding remarks}

We have presented a method to compute all-vs-all MEMs that rely on grammar compression to reduce memory overhead and save redundant calculations. In particular, given a collection $\mathcal{T}$ of $n$ symbols and $occ$ MEMs, we can get a grammar $\mathcal{G}$ of size $G$ in $O(n)$ time and $O(G\log G)$ bits of space, and find the MEMs of $\mathcal{T}$ on top of $\mathcal{G}$ in $O(G+occ)$ time and using $O(\log G(G+occ))$ bits. We believe our framework is of practical interest as it uses mostly plain data structures that store satellite data about the grammar. Besides, we can either compute the MEMs at the same time we construct the grammar or do it later. However, it remains open to check how far is $G$ from $O(\gamma \log \frac{n}{\gamma})$. The comparison is reasonable as \textsc{FFGram}, our grammar algorithm, resembles the locally-consistent grammar of Christiansen et al.~\cite{christiansen2020optimal}, which achieves that bound. Still, it is unclear to us how the overlap produced by \textsc{FFGram} affects the bound. Overlapping phrases has an exponential effect on the grammar size but also affects how the text is parsed. An interesting idea would be to find a way to chain prMEMs as approximate matches and then report those in the last step of our algorithm instead of the MEMs. An efficient implementation of such a procedure could significantly reduce the cost of biological sequence analyses in massive collections. 

\bibliography{references}
\appendix
\clearpage\section*{Appendix}

\section{Figures and Examples}

\begin{figure}[!htp]
\centering
\includegraphics[width=\textwidth]{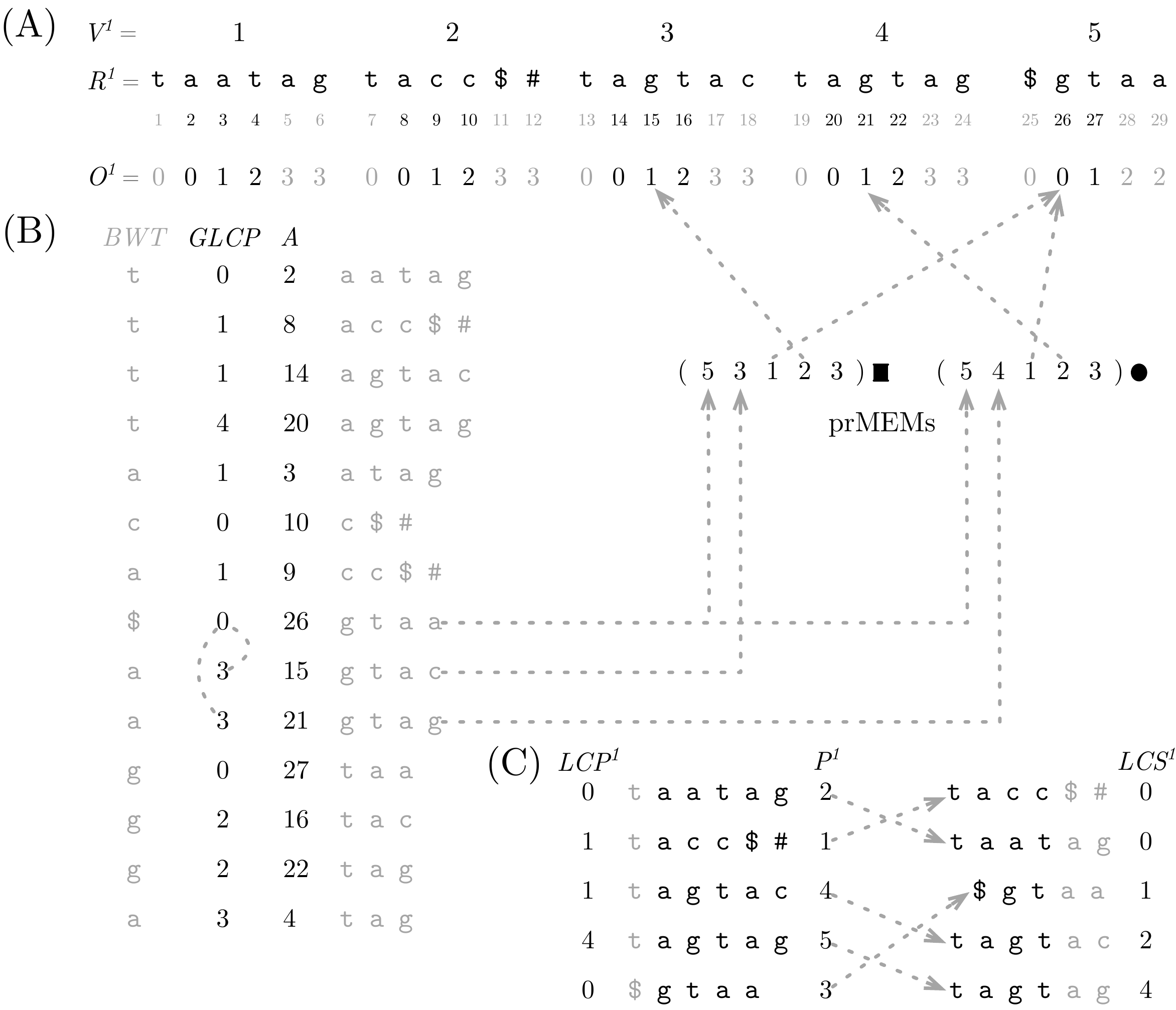}
\caption{First iteration to find prMEMs with $\tau=3$ in the grammar of Figure~\ref{fig:ovp_gramp}. (A) The encoding of $\mathcal{R}^{i}$. Numbers between $R^{1}$ and $O^{i}$ are indexes. Grey indexes denote the suffixes of $R^{i}$ that are not included in the sparse suffix array $A$. (B) Detection of prMEMs in the grammar level $1$. The dashed lines in $GLCP$ indicate the MEMs found by the suffix-tree-based MEM algorithm, and the dashed arrows map the prMEMs to their corresponding prMEM tuples. The $BWT$ shows the left context of the suffixes. (C) The satellite data structures for the next iteration. The symbols in grey are those that are ignored in the string sorting}
\label{fig:prmem_it1}
\end{figure}

\begin{example}
Computing prMEMs in level 1 (Figure~\ref{fig:prmem_it1}) of the grammar of Figure~\ref{fig:ovp_gramp}. The simulation of the suffix-tree-based algorithm reported the MEMs $(A[8]=26, A[9]=15, 3)$ and $(A[8]=26, A[10]=21, 3)$ in $R^{1}$ (dashed lines in $GLCP$). These pairs correspond to the occurrences $R^1[26]$, $R^{1}[15]$ and $R^{1}[21]$ of \texttt{gta}. The operation $parent$ gives the nonterminals $5=parent(26)$, $3=parent(15)$ and $4=parent(21)$, which are stored in the prMEM tuples and marked with dashed arrows. We compute the offsets $1,2$ within the left prMEM tuple $(5,3,1,2,3)$ by accessing the positions $1=O[26]+1$ and $2=O[15]+1$, respectively. These offsets are the relative positions of the MEM within $efexp(5)$ and $efexp(3)$, respectively. The last element of the tuple, $3$, comes from the triplet reported by the suffix-tree-based MEM algorithm and corresponds to the length $GLCP[9]=3=|\texttt{gta}|$ of the MEM. We compute the information of the right prMEM tuple $(5,4,1,2,3)$ in the same way but using its corresponding positions $A[8]=26$ and $A[10]=21$.
\end{example}

\begin{figure}[!htp]
\centering
\includegraphics[width=\textwidth]{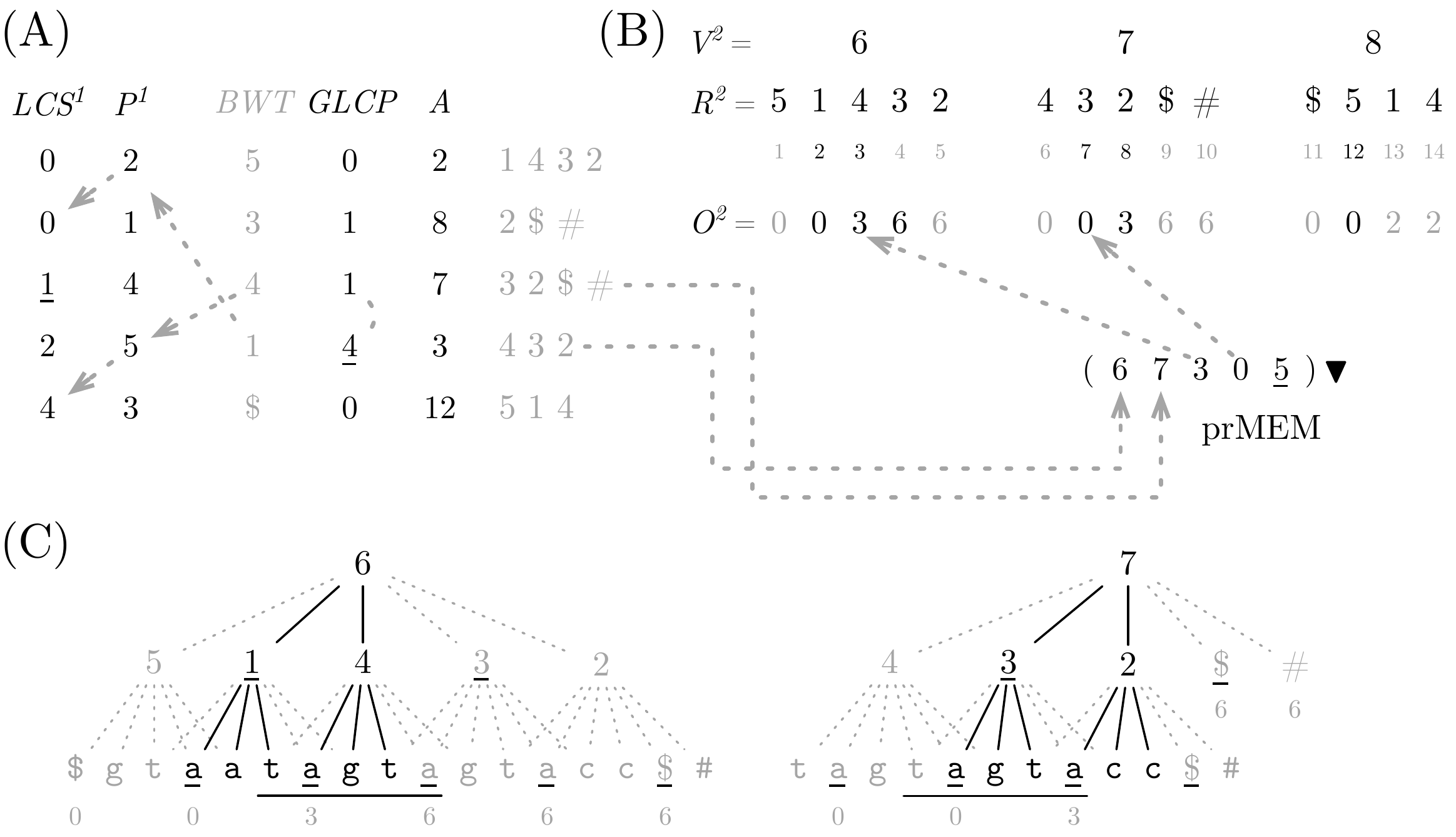}
\caption{Second iteration to find prMEMs with $\tau \geq 3$ in the grammar of Figure~\ref{fig:ovp_gramp}. (A) Data structures to compute the prMEMs. The $GLCP$ values are already transformed to encode the longest common prefixes between $rexp$ expansions of the $R^{i}$ suffixes in $A$. The $BWT$ column shows the left context of each $R^{i}$ suffix in $A$. The dashed arrows indicate the steps to get the left boundary of the prMEM. (B) Encoding of $R^{2}$ and $O^{2}$. The indexes in grey indicate the suffixes of $R^{i}$ that are not in $A$. (C) Offset sequence in $O$ for the nonterminals $6$ and $7$.}
\label{fig:prmem_it2}
\end{figure}

\begin{example}
Computing prMEMs in level 2 (Figure~\ref{fig:prmem_it2}) of the grammar of Figure~\ref{fig:ovp_gramp}. The suffix-tree-based algorithm reported the triplet $(A[3]=7, A[4]=3, 4)$ as a MEM (dashed line in $GLCP$). The corresponding suffixes $R^{i}[A[3]=7..]=32\$\#$ and $R^{i}[A[4]=3..]=432$ do not share a prefix, but $GLCP$ indicates their $rexp$ expansions $rexp(32\$\#)=\texttt{agtacc\$\#}$ and $rexp(432)=\texttt{agtagtacc\$\#}$ do share a prefix of length $GLCP[4]=4$. Still, the MEM is incomplete as the match can be extended to the left. To find the left-extension, we access the symbols $R^{2}[A[3]-1=6]=4$ and $R^{2}[A[4]-1=2]=1$ and compute the longest common suffix of their $lexp$ expansions. This step require us to visit $P^{1}[4]=5$ and $P^{1}[1]=2$, and compute $lcs^{1}(4, 1) = rmq(LCS^{1}, 2, 5)=1$ (underlined value in $LCS^{1}$). We add $1$ to the current MEM length to get the final length $\ell=4+1=5$ (rightmost value in the prMEM tuple). We compute $6$ and $7$ in the prMEM tuple as before, i.e., via $parent$. The values $3$ and $0$ indicate the positions of the MEM within $efexp(6)$ and $efexp(7)$. As before, we use $O$ to get these values but also subtract the length of the left extension. Thus, $3=O[A[3]]-1+1$ and $0=O[A[4]]-1+1$. Notice that the offset $0\leq1$ indicates that the MEM starts before the relative position of $efexp(7)$ within $lexp(7)$. Figure~\ref{fig:prmem_it2}C depicts this idea about the offset. 
\end{example}



\section{Simplifying the grammar}\label{sec:simp_gram}

The simplification consists in recursively removing from $R$ each nonterminal $X \in V$ that has one occurrence and leaving just its replacement. Suppose the left-hand side of $X \rightarrow F$ occurs only once in $R$ in the string for the rule $X' \rightarrow AXZ$. Then, we remove $F$ from $R$ and reinsert it in the area of $AXZ$ to replace the symbol $X$ (i.e., the right-hand side for $X'$ becomes $AFZ$). This process is recursive because we also simplify the sequence of $F$. Additionally, we update the $efexp$ expansions stored in $O$. The idea is similar to what we did in $R$: the value $e_X$ in $X' \rightarrow e_Ae_Xe_Z \in \mathcal{O}$ associated with $X$'s occurrence in $X' \rightarrow AXZ$ is now replaced by the offset sequence in $O$ for the string $F$. We also need to update $F$'s offset sequence by adding $efexp(A)$ to its values. The simplification requires renaming the nonterminals so they form a contiguous range of values again.

We keep track of the changes we made during the simplification to update the values in $\mathcal{L}$. In particular, we create a vector $K[1..g]$ storing for each $X \rightarrow F$ the nonterminal $K[X]=X' \in V$ that now encloses the simplified version of $F$. Let $X' \rightarrow A_sF_sZ_s$, with $A_s,F_s,Z_s \in (\Sigma \cup V)^{*}$ be the simplified rule for $X'$. Another vector $E$ stores in $E[X]=|A_s|$ the number of symbols in the prefix $A_s$ preceding $F_s$. Now we scan $\mathcal{L}$, and for each tuple $(X, Y, o_X, o_Y, \ell)$, we update the symbols $X=K[X], Y=K[Y]$, and their corresponding offsets $(o_X=o_X+O[E[X]], o_Y=o_Y+O[E[Y]])$. After simplifying the grammar and update $\mathcal{L}$, we get rid of $E$ and $K$.

\end{document}